\documentclass[runningheads]{llncs}

\usepackage{todonotes}
\usepackage{amssymb}
\usepackage{amsmath}

\def\comic#1#2#3{\parbox{#1}{\centering\includegraphics[width=#1]{#2}\\{\footnotesize #3}}}
\def\comicII#1#2{\parbox{#1}{\centering\includegraphics[width=#1]{#2}}}

\newcommand{\kVR}{k\mbox{VR}}
\newcommand{\tVR}{2\mbox{VR}}
\newcommand{\zVR}{0\mbox{VR}}

\newtheorem{observation}{Observation}
\title{$k$-Transmitter Watchman Routes\thanks{Supported by grants 2018-04001 (Nya paradigmer f\"{o}r autonom obemannad flygledning) and 2021-03810 (Illuminate: bevisbart goda algoritmer f\"{o}r bevakningsproblem) from the Swedish Research Council (Vetenskapsr\r{a}det).}}

\author{Bengt J. Nilsson\inst{1}\orcidID{0000-0002-1342-8618} \and
Christiane Schmidt\inst{2}\orcidID{0000-0003-2548-5756} }
\authorrunning{}
\institute{Department of Computer Science and Media Technology, Malm\"o University, Sweden, \texttt{ bengt.nilsson.TS@mau.se} \and
Department of Science and Technology, Link\"oping University, Sweden\\
  \texttt{christiane.schmidt@liu.se}}
\begin{document}
\maketitle              
\begin{abstract}We consider the watchman route problem for a $k$-transmitter watchman: standing at point $p$ in a polygon $P$, the watchman can see $q\in P$ if $\overline{pq}$ intersects $P$'s boundary at most $k$ times---$q$ is $k$-visible to $p$. Traveling along the $k$-transmitter watchman route, either all points in $P$ or a discrete set of points $S\subset P$ must be $k$-visible to the watchman. We aim for minimizing the length of the $k$-transmitter watchman route.

We show that even in simple polygons the shortest $k$-transmitter watchman route problem for a discrete set of points $S\subset P$ is NP-complete and cannot be approximated to within a logarithmic factor 
(unless P=NP), both with and without a given starting point. Moreover, we present a  polylogarithmic approximation for the $k$-transmitter watchman route problem for a given starting point and $S\subset P$ with approximation ratio $O(\log^2(|S|\cdot n) \log\log (|S|\cdot n) \log|S|)$ (with $|P|=n$).

\keywords{Watchman Route  \and $k$-Transmitter \and $k$-Transmitter Watchman Route \and NP-Hardness \and Approximation Algorithm \and NP-completeness}
\end{abstract}

\section{Introduction}\label{sec:intro}

In the classical {\it Watchman Route Problem} (WRP)---introduced by Chin and Ntafos~\cite{cn-owr-86}, we ask for the shortest (closed) route in an environment (usually a polygon $P$), such that a mobile guard traveling along this route sees all points of the environment. The WRP has mostly been studied for the ``traditional'' definition of visibility: a point $p\in P$ sees another point $q\in P$ if the line segment $\overline{pq}$ is fully contained in $P$. This mimics human vision, as this, e.g., does not allow looking through obstacles or around corners. 
In contrast to the classical guarding problems with stationary guards, the {\it Art Gallery Problem} (AGP)---where we aim to place a minimum number of non-moving guards that see the complete environment---the WRP is solvable in polynomial time in simple polygons with~\cite{cn-swrsp-91,thi-ciacs-99,defm-tsp-03} 
and without~\cite{cjn-fswrsp-93,t-fcswr-01} a given boundary start point. In polygons with holes, the WRP is NP-hard~\cite{cn-owr-86,dt-wtph-12}.

However, we may also have other vision types for the watchman. If we, for example, consider a mobile robot equipped with a laser scanner, then the scanning creates point clouds, which are easier to map afterwards when the robot was immobile while taking a single scan. This results in the model of ``discrete vision'': information on the environment can be acquired only at discrete points, at other times the watchman is blind. 
Carlsson et al.~\cite{cnn-ogcmw-93} showed that the problem of finding the minimum number of vision points---the discrete set of points at which the vision system is active---along a given path (e.g., the shortest watchman route) is NP-hard in simple polygons. Carlsson et al.~\cite{carlsson99computing} also presented an
efficient algorithm to solve the problem of placing vision points along a given watchman route in streets.
Another natural restriction for a mobile robot equipped with laser scanners is a limited visibility range (resolution degrades with increasing distance), see~\cite{bgs-euped-01,fms-mctc-10,s-amalc-11}.

Another type of visibility is motivated by modems: When we try to connect to a modem, we observe that one wall will not prevent this connection (i.e., obstacles are not always a problem). However, many walls separating our location from the modem result in a failed connection. This motivates studying so-called $k$-transmitters: $p\in P$ sees $q\in P$ if the line segment $\overline{pq}$ intersects $P$'s boundary at most $k$ times. If more than $k$ walls are intersected, we no longer ``see'' an object---the connection is not established. Different aspects of guarding with $k$-transmitters (the AGP with $k$-transmitters) have been studied. First, the focus was on worst-case bounds, so-called Art Gallery theorems. Aichholzer et al.~\cite{affhhuv-mimp-18} presented tight bounds on the number of $k$-transmitters in monotone and monotone orthogonal polygons. Other authors explored $k$-transmitter coverage of regions other than simple polygons, such as 
coverage of the plane in the presence of lines or line segment obstacles 
\cite{bbal-cktpo-10,mvu-mip-09}. Ballinger et al.~\cite{bbal-cktpo-10} also presented a tight bound for a very special class of polygons: spiral polygons, so-called {\it spirangles}. Moreover, for simple $n$-gons they provided a lower bound of $\lfloor n/6 \rfloor$ 2-transmitters. Cannon et al.~\cite{cfils-ccgpe-18} showed that it is NP-hard to compute a minimum cover of point 2-transmitters, point $k$-transmitters, and edge 2-transmitters (where a guard is considered to be the complete edge) in a simple polygon. The point 2-transmitter result extends to orthogonal polygons.
Moreover, they gave 
upper and lower bounds for the number of edge 2-transmitters in general, monotone, orthogonal monotone, and orthogonal polygons; and improved the bound from \cite{bbal-cktpo-10} 
for simple $n$-gons to $\lfloor n/5 \rfloor$ 2-transmitters. For the AGP with $k$-transmitters, no approximation algorithms have been obtained so far, but Biedl et al.~\cite{bclmm-goags-19} recently presented a first constant-factor approximation result for so-called sliding $k$-transmitters (traveling along an axis-parallel line segment $s$ in the polygon, covering all points $p$ of the polygon for which the perpendicular from $p$ onto $s$ intersects at most $k$ edges of the polygon).

Of course, $k$-transmitters do not have to be stationary (or restricted to travel along a special structure as in~\cite{bclmm-goags-19}): we might have to find a shortest tour such that a mobile $k$-transmitter traveling along this route can establish a connection with all (or a discrete subset of the) points of an environment, the {\em WRP with a $k$-transmitter}. This problem is the focus of this paper and to the best of our knowledge it has not been studied before. Given that our watchman moves inside the polygon, we consider even values for $k$ only---while odd numbers of $k$ can be interesting when we, for example, want to monitor the plane in presence of line-segment or line obstacles, or when we want to monitor parts of a polygon's exterior.

For the original WRP, we know that an optimal tour must visit all essential cuts: the non-dominated extensions of edges incident to a reflex vertex. However, already if we want to see a discrete set of points with a mobile $k$-transmitter, for $k\geq2$, we do not have such a structure: the region visible to a $k$-transmitter point can have $O(n)$ connected components, to see the point, the mobile $k$-transmitter can visit any of these. 

Guarding a discrete set of points---though with stationary guards---is, e.g., considered in the problem of guarding treasures in an art gallery: Deneen and Joshi~\cite{dj-tag-92} presented an approximation algorithm for finding the minimum number of guards that monitors a discrete set of treasure points, Carlsson and 
Jonsson~\cite{cj-gt-93} added weights to the treasures and aimed for placing a single guard maximizing the total value of the guarded treasures.

{\bf Roadmap.} In Section~\ref{sec:not}, we introduce notation; in Section~\ref{sec:sp}, we detail some special properties of $k$-transmitters. In Section~\ref{sec:np}, we show that the WRP with $k$-transmitters monitoring a discrete set of points is NP-complete and cannot be approximated to within a logarithmic factor in simple polygons even for $k=2$.  
In Section~\ref{sec:appx}, we present an approximation algorithm for the 
WRP with $k$-transmitters monitoring a discrete set of points and has a given starting point.

\section{Notation and Preliminaries}\label{sec:not}

We let $P$ be a polygon, in general, we are interested in $P$ being simple. We define $\partial(P)$ as the boundary of $P$, and let $n$ denote the number of vertices of $P$.

A point $q\in P$ is {\it $k$-visible} to a $k$-transmitter $p\in \mathbb{R}^2$ if $\overline{qp}$ intersects $P$'s boundary in at most $k$ connected components. This includes ``normal'' guards for $k=0$.  For a point $p\in P$, we define the {\it $k$-visibility region} of $p$, $\kVR(p)$, as the set of points in $P$ that are $k$-visible from $p$, see Figure~\ref{fig:boundary-is-not-enough}(a). For a set $X\subseteq P$: $\kVR(X)=\bigcup_{p\in X} \kVR(p)$. A $k$-visibility region can have $O(n)$ connected components (CCs), see~\cite{cfils-ccgpe-18} and Observation~\ref{obs:cc} in Section~\ref{sec:sp}, we denote these components by $\kVR^j(p), j=1,\ldots, J_p$, with $J_p\in O(n)$.

The boundary of each CC of $\kVR(p)$ contains edges that coincide with (parts of) edges of $\partial(P)$, and so-called {\it windows}. Cutting $P$ along a window $w$ partitions it into two subpolygons. We denote by $P_s(w)$ the subpolygon that contains a given point $s\not\in w$ and consider the window $w$ to belong to $\partial\big(P_s(w)\big)$. A window $w_1$ {\it dominates} another window $w_2$ if $P_s(w_2)\subset P_s(w_1)$. A window $w$ is {\it essential} if it is not dominated by any other window.

For a given point $s$ in a simple polygon $P$ there exists one window $w$ per CC $\kVR^j(p)$, such that any path from $s$ to a window $w'\neq w, w'\in \kVR^j(p)$ intersects $w$, that is, $P_s(w)\subseteq P_s(w')$. We denote this window as the {\it cut} of $\kVR^j(p)$ w.r.t. $s$, see Figure~\ref{fig:boundary-is-not-enough}(a). 


We aim to find shortest watchman routes for $k$-transmitters. In particular, we aim to find a route $R$, such that either all points of a polygon $P$ or a set of points $S\subset P$ is $k$-visible for the 
watchman following $R$, that is, $\kVR(R)=P$ or $S\subset\kVR(R)$.
We define the {\it k-Transmitter WRP for $X\subseteq P$ and $P$}, $k$-TrWRP($X,P$), possibly with a given starting point $s\in P$, $k$-TrWRP($X,P,s$), as the problem of finding the shortest route for a $k$-transmitter within $P$, starting at $s$, from which every point in $X$ is $k$-visible. 
Let OPT($S,P$) and OPT($S,P,s$) be {\it optimal} w.r.t. $k$-TrWRP($S,P$) and $k$-TrWRP($S,P,s$), respectively. 

In Section~\ref{sec:appx}, we use an approximation algorithm by Garg et al.~\cite{gkr-paags-00} for the group Steiner tree problem. The {\it group Steiner tree problem} was introduced by Reich and Widmayer~\cite{rw-bspvl-89}: given a graph $G=(V,E)$ with cost function $c:E\rightarrow \mathbb{R}^+$ and subsets of vertices $\gamma_1, \gamma_2, \ldots, \gamma_Q\subseteq V$\!, so-called {\it groups}, we aim to find the minimum-cost subtree $T$ of $G$ that contains at least one vertex from each of the groups, that is, a connected subgraph $T=(V',E'), V'\subseteq V, E'\subseteq E$ that minimizes $\sum_{e\in E'} c_e$ such that $V'\cap\gamma_q\neq\emptyset,\; \forall q\in\{1,\ldots,Q\}$. For $|V|=m$, 
Garg et al.~\cite{gkr-paags-00} obtained a randomized algorithm with an approximation ratio of $O(\log^2 m \log\log m \log Q)$.

\section{Special Observations for $k$-Transmitters}\label{sec:sp}

\begin{figure}[t]
\centering
\comic{.3\textwidth}{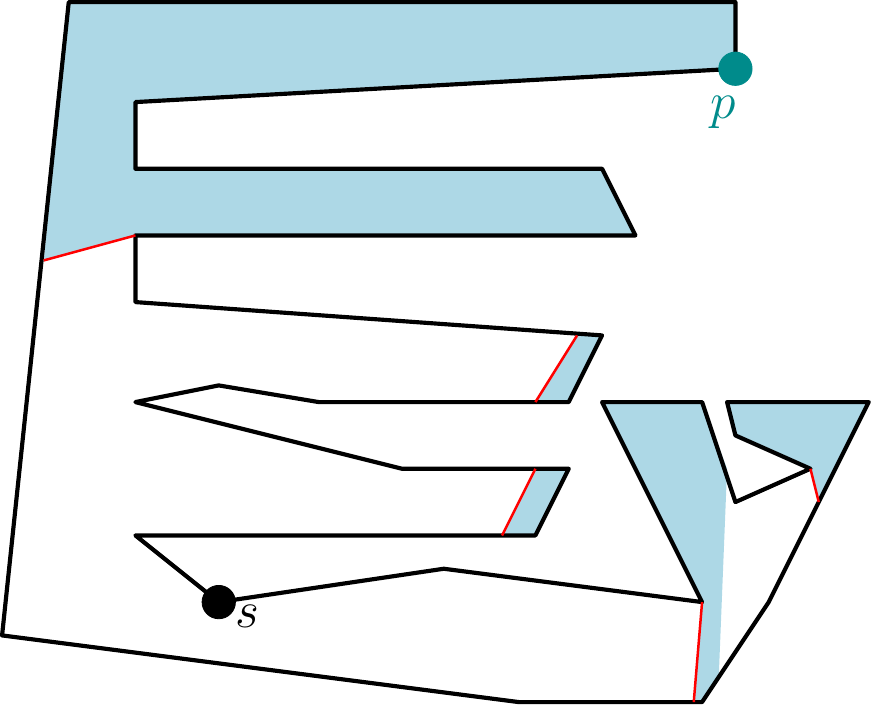}{(a)}\hfill
\comic{.45\textwidth}{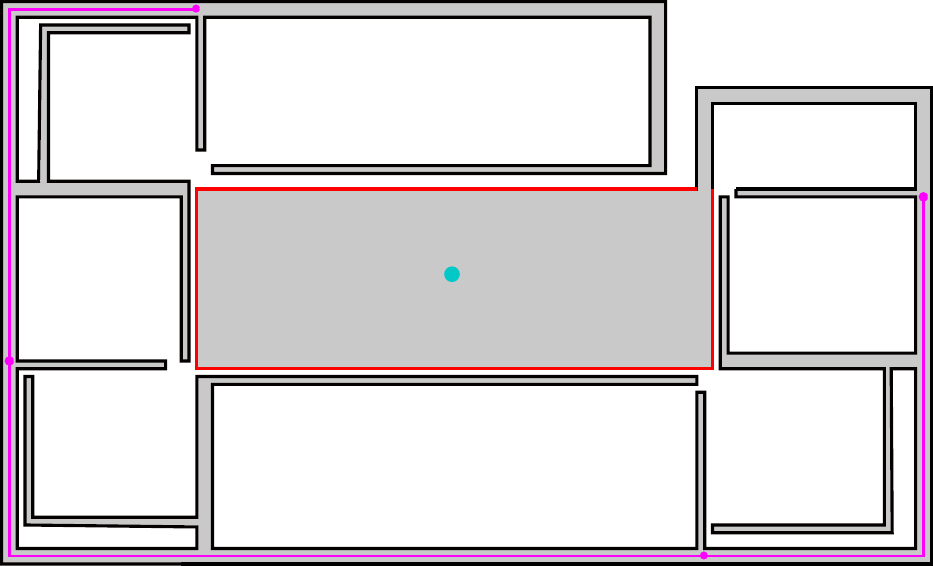}{(b)}
  \caption{\small (a): Point $p$ with its $2$-visibility region shown in light blue, $\tVR(p)$ has five CCs. The cuts of these CCs w.r.t. $s$ are shown in red. (b): The complete boundary of this polygon is visible from the pink $2$-transmitter watchman route. In particular, this holds for the red part of the polygon's boundary (seen, e.g., from the four marked pink points). However, the turquoise point is not $2$-visible from that route. Thus, not all of $P$ is $2$-visible from that route.}
  \label{fig:boundary-is-not-enough}
\end{figure}

For $0$-transmitter watchmen guarding a simple polygon's boundary is enough to guard all of the polygon, this does not hold for $k$-transmitters with $k\geq 2$:
\begin{observation}
For a simple polygon $P$ and $k\geq 2$: $\partial(P)$ being $k$-visible to a $k$-transmitter watchman route is not a sufficient condition for $P$ being $k$-visible to that $k$-transmitter watchman route, see Figure~\ref{fig:boundary-is-not-enough}(b).
\end{observation}

The visibility region $\zVR(p)$ for any point $p\in P$ has a single connected component (and is also a simple polygon). This does not hold for larger $k$, as already for $k=2$ we have:
\begin{observation}[Observation 1 in~\cite{cfils-ccgpe-18}]\label{obs:cc}
In a simple polygon $P$, the $2$-visibility region of a single guard
can have $O(n)$ connected components.  [More precise: The $2$-visibility region of a single guard can have at most $n$ connected components.]
\end{observation}

\section{Computational Complexity}\label{sec:np}

\begin{theorem}\label{np-disc}
For a discrete set of points $S\subset P$ and a simple polygon $P$, the $k$-Transmitter WRP for $S$ and $P$,
$k$-TrWRP($S,P$), does not admit a polynomial-time approximation algorithm with approximation ratio $\alpha\cdot\ln |S|$ for a constant $\alpha>0$ unless P=NP, even for $k=2$. 
\end{theorem}

\begin{proof}
We give a gap-preserving reduction from Set Cover (SC):

\noindent{\bf Set Cover (SC):}\\
{\bf Input:} A set system $(\mathcal{U},\mathcal{C})$, 
with $\cup_{C\in\mathcal{C}} C = \mathcal{U}$.\\
{\bf Output:} Minimum cardinality sub-family $\mathcal{B}\subseteq\mathcal{C}$ that covers $\mathcal{U}$, i.e., $\cup_{B\in\mathcal{B}} B = \mathcal{U}$.

Given an instance of the Set Cover problem, we construct a polygon $P$ with $S=\mathcal{U}\cup \{v\}$. For that construction, we build a bipartite graph $G$ with vertex set $V(G)=\mathcal{U}\cup\mathcal{C}$ and edge set $E(G)=\{e=(u,c)\mid u\in\mathcal{U}, c\in\mathcal{C}, u \in c\}$. See Figure~\ref{fig:sc-red}(a) for an example of this graph~$G$. 

We start the construction of $P$, see Figure~\ref{fig:sc-red}(b), with a spiral structure with $v\in S$ located in the center of the spiral, to its end we attach $|\mathcal{C}|$ spikes (narrow polygonal corridors of four vertices each),  each ending at the same $y$-coordinate (each $C\in\mathcal{C}$ corresponds to the tip of one spike). Let the length of the longest spike be $\ell_\mathcal{C}$, and let the length of the spikes differ by $\varepsilon'\ll\ell_\mathcal{C}$ only.
All points $u\in\mathcal{U}$ are located in a long horizontal box to which T-shaped structures are attached, such that the crossbeams leave gaps only where an edge in $E(G)$ connects a $C$ from a spike to a $u$ in the horizontal box. These two polygon parts are connected by a very long vertical polygonal corridor. Let the length of this corridor be $\ell_{\mbox{vert}} = 4\cdot |\mathcal{C}| \cdot\ell_\mathcal{C}$. 

\begin{figure}[h]
\centering
\comic{.2\textwidth}{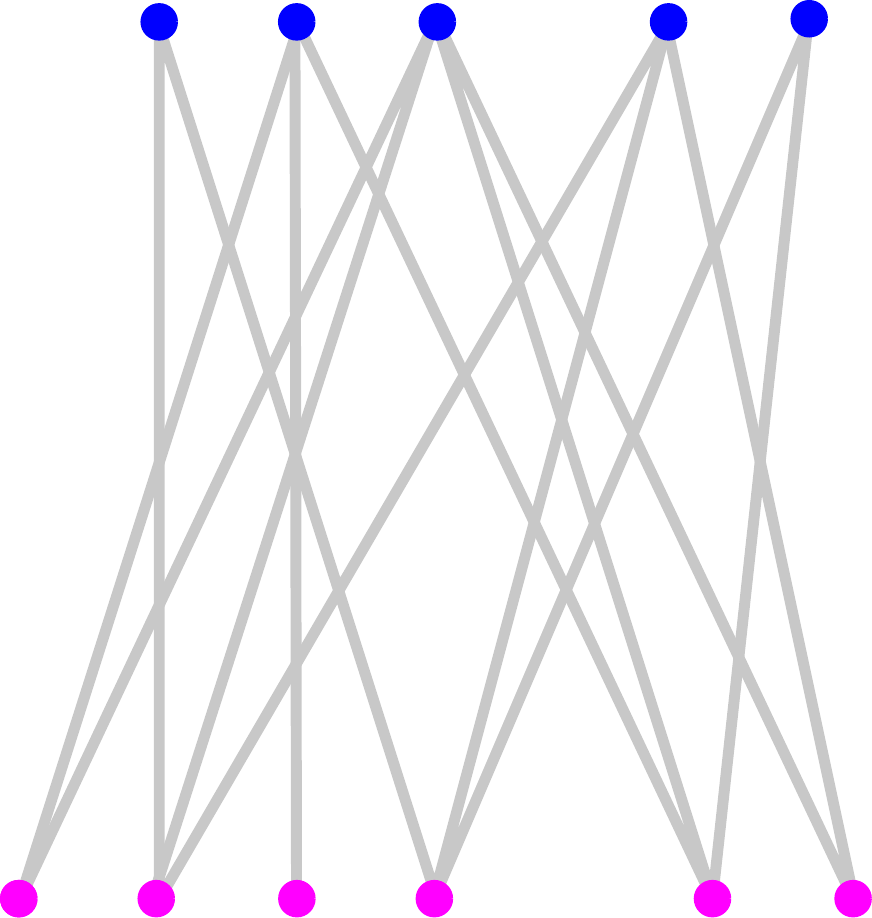}{(a)}\hfill
\comic{.35\textwidth}{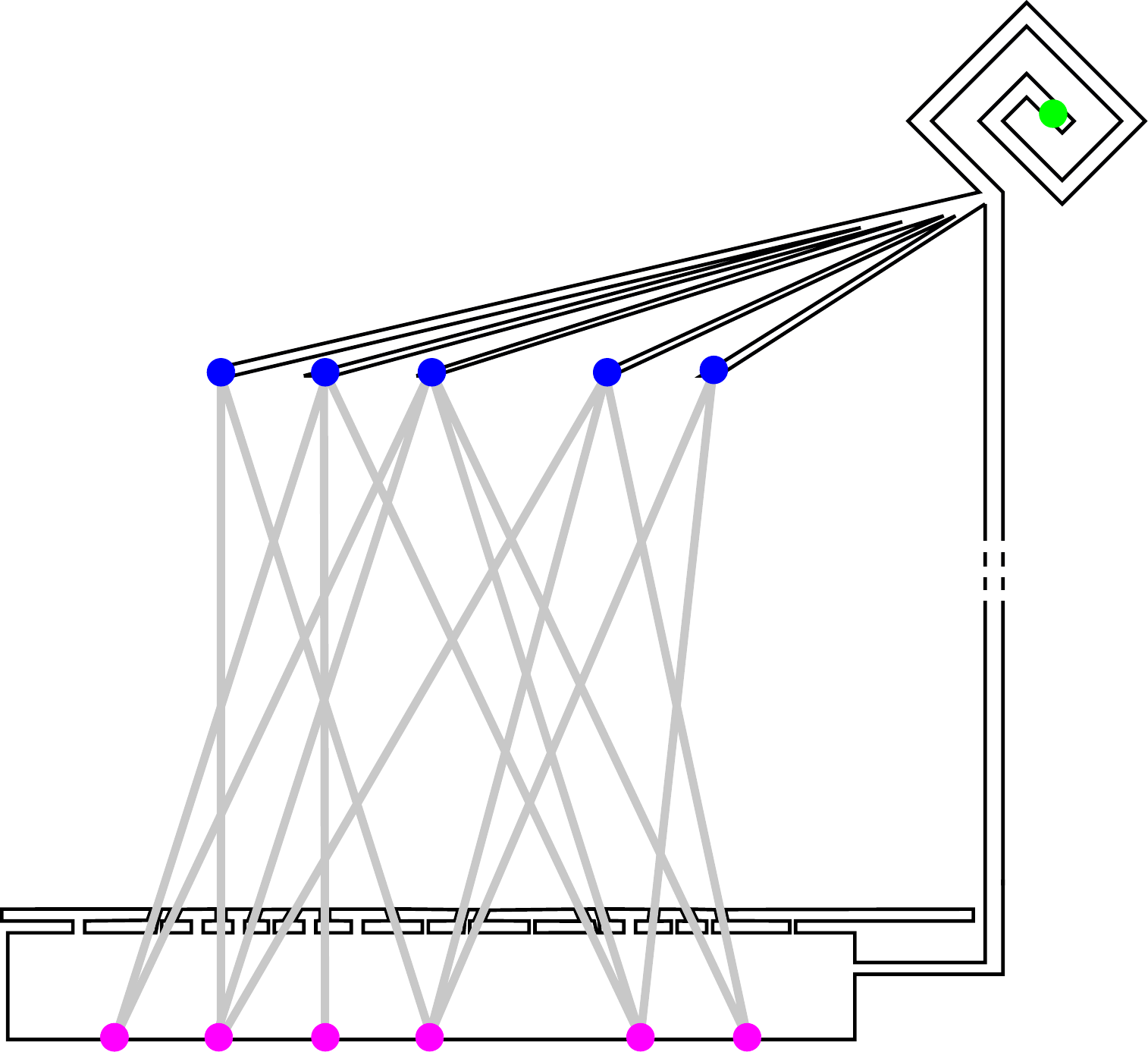}{(b)}\\
\comic{.35\textwidth}{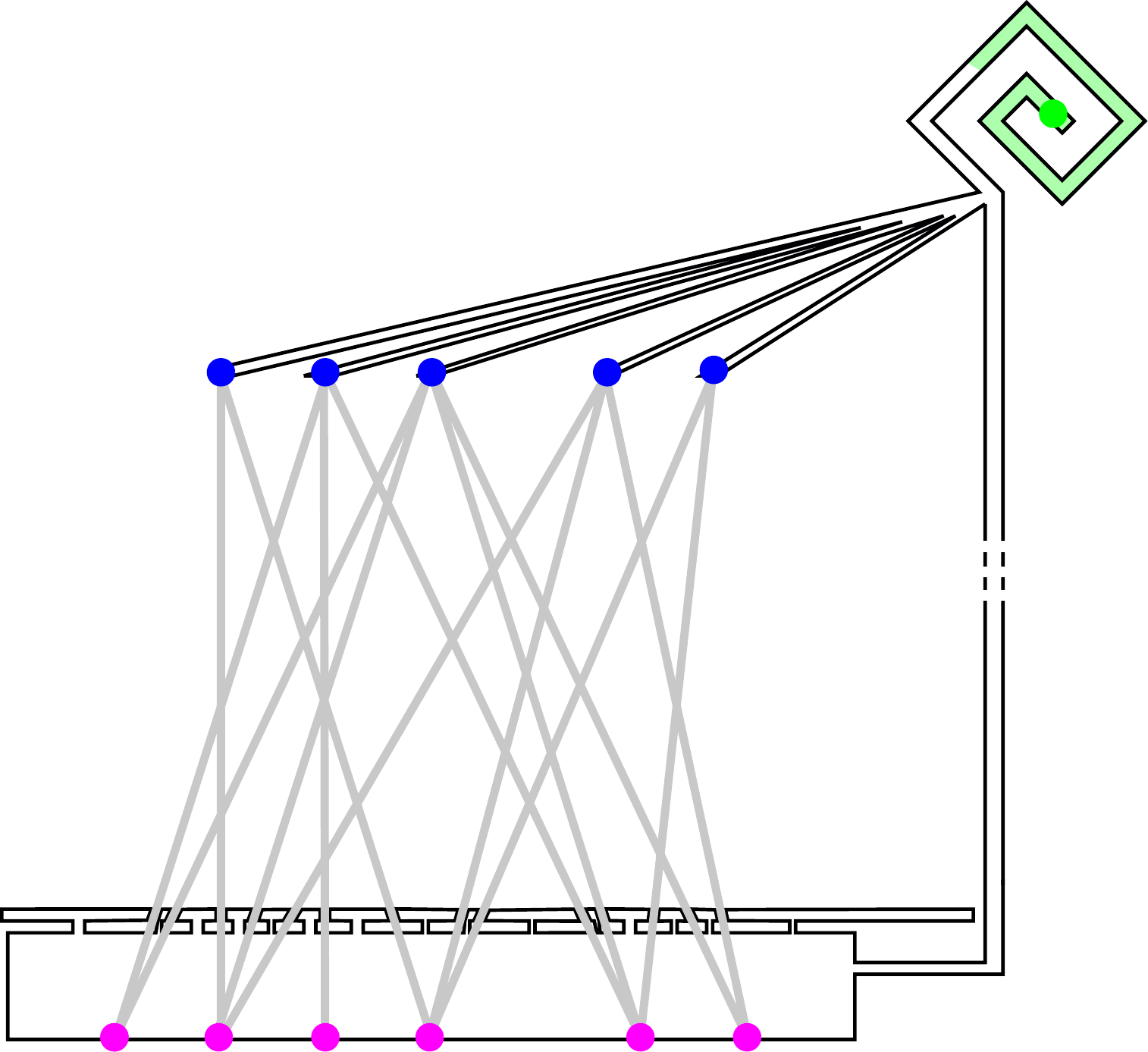}{(c)}\hfill
\comic{.35\textwidth}{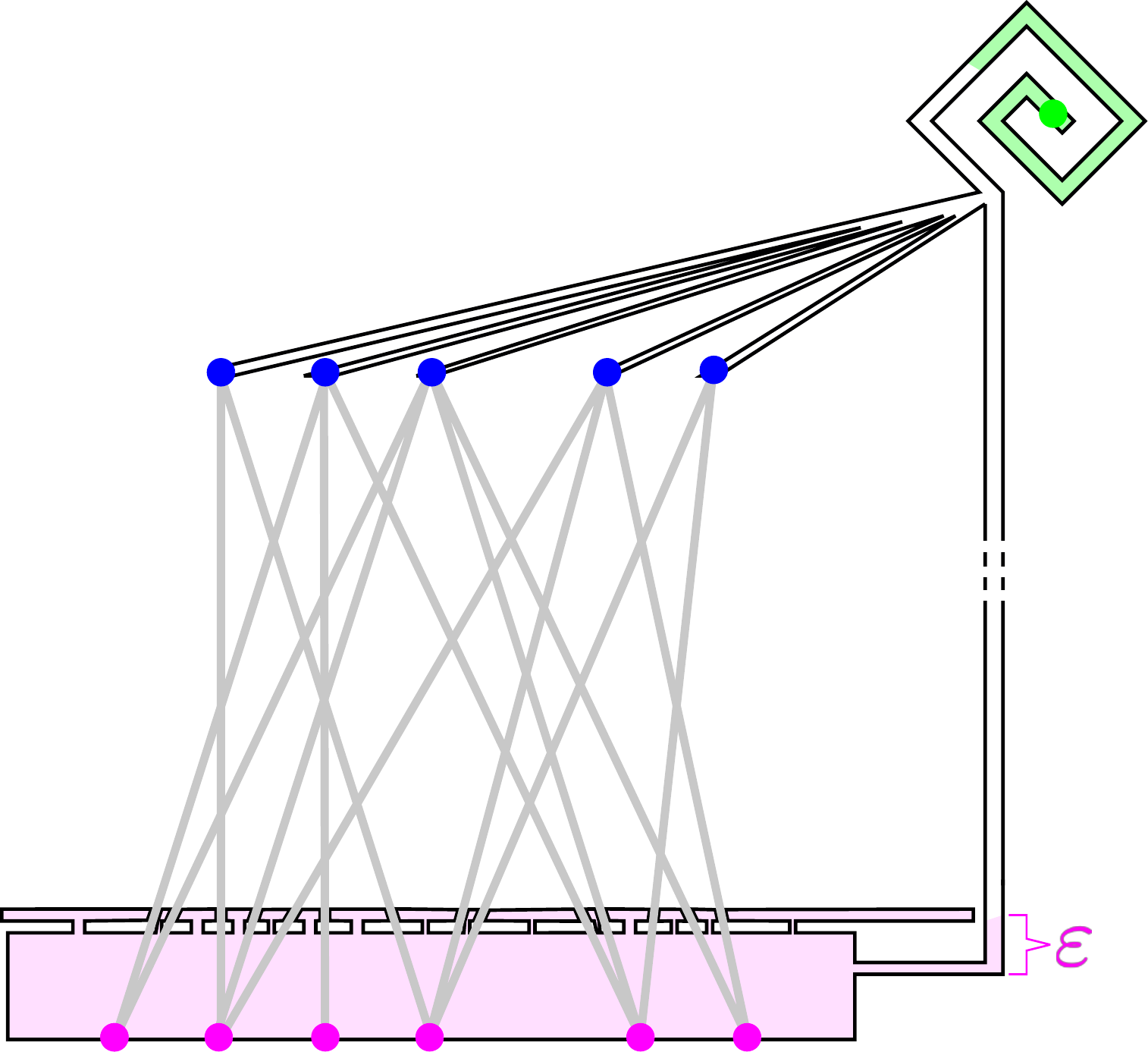}{(d)}
  \caption{\small Example construction for the SC instance $(\textcolor{magenta}{\mathcal{U}},\textcolor{blue}{\mathcal{C}})$ with \textcolor{magenta}{$\mathcal{U}=\{1,2,3,4,5,6\}$}, \textcolor{blue}{$\mathcal{C}=\{\{2,4\},\{1,3,5\},\{1,2,5,6\},\{2,4,6\},\{4,5\}\}$}. (a) Graph $G$, (b) polygon $P$ with \textcolor{green}{$v$} shown in green, (c) $\kVR(v)$ shown in light green, (d) set of points that are not located in the $|\mathcal{C}|$ spikes and see all $u\in\mathcal{U}$ (all points in the horizontal box, in the T-shaped structures,  and points at the bottom of the long vertical corridor).}
  \label{fig:sc-red}
\end{figure}

Because of the placement of $v$, any 2-transmitter needs to enter the spiral structure to reach a point  in $\tVR(v)$ (indicated in light green in Figure~\ref{fig:sc-red}(c)). All $u\in\mathcal{U}$ are visible only to points in the horizontal box, in the T-shaped structures, to points at the bottom of the long vertical corridor (shown in light pink in Figure~\ref{fig:sc-red}(d)) and from the tips of the spikes (representing the $C\in\mathcal{C}$). Covering the vertical corridor twice to reach any of the light pink points from $\tVR(v)$ is more expensive than even visiting all tips of the spikes. Hence, any optimal $k$-transmitter watchman route must visit spike tips to see all $u\in\mathcal{U}$. To obtain the shortest $k$-transmitter watchman route we must visit as few spike tips as possible: we must visit the minimum number of spike tips, such that all pink points $u\in\mathcal{U}$ are covered. This is exactly the solution to the Set Cover problem.

Set Cover cannot be approximated in polynomial time to within a factor $(1-o(1))\ln |\mathcal{U}|$, where $|S|=|\mathcal{U}|+1$;~\cite{f-tasc-98}.

For each $s_i\in S$, we can compute the $2$-visibility region of $s_i$ and check whether the given route intersects it, thus, $k$-TrWRP($S,P$) is in~NP.
\end{proof}


By choosing $s$ to be located on the window of $\tVR(v)$ in the above construction and using $S=\mathcal{U}$, we obtain:
\begin{corollary}
For $S$, $P$, and $\alpha$ as in Theorem~\ref{np-disc}, $s\in P$,
{\rm $k$-TrWRP($S,P,s$)} does not admit a 
polynomial-time 
approximation algorithm with approximation ratio $\alpha\cdot\ln |S|$, 
for $k\geq2$.
\end{corollary}

We can generalize the construction by replacing the elements of $U$ in Fig.~\ref{fig:sc-red} with small almost horizontal spikes that need to be covered by the tour visiting the minimum number of set spikes at the top of the figures. We claim: 
\begin{corollary}
For a simple polygon $P$ and $s\in P$,
{\rm $k$-TrWRP($P,P,s$)} does not admit a 
polynomial-time 
approximation algorithm with approximation ratio $\alpha\cdot\ln n$ for a constant $\alpha>0$, 
for $k\geq4$.
\end{corollary}

\section{Approximation Algorithm for $k$-TrWRP($S,P,s$)}\label{sec:appx}

In this section, we develop an approximation algorithm ALG($S,P,s$) for the $k$-transmitter watchman route problem for a simple polygon $P$, a discrete set $S$ of points in $P$, and a given starting point $s$. We prove: 
\begin{theorem}\label{th:appx-dis}
Let $P$ be a simple polygon, $n=|P|$. Let {\rm OPT($S,P,s$)} be the optimal solution for the {\rm $k$-TrWRP($S,P,s$)} and let $R$ be the solution output by our algorithm {\rm ALG($S,P,s$)}. 
Then $R$ has length within 
$O\big(\log^2 (|S|\cdot n) \log\log (|S|\cdot n) \log|S|\big)$ of OPT($S,P,s$). 
\end{theorem}

\begin{figure}   
\centering
\comic{.48\textwidth}{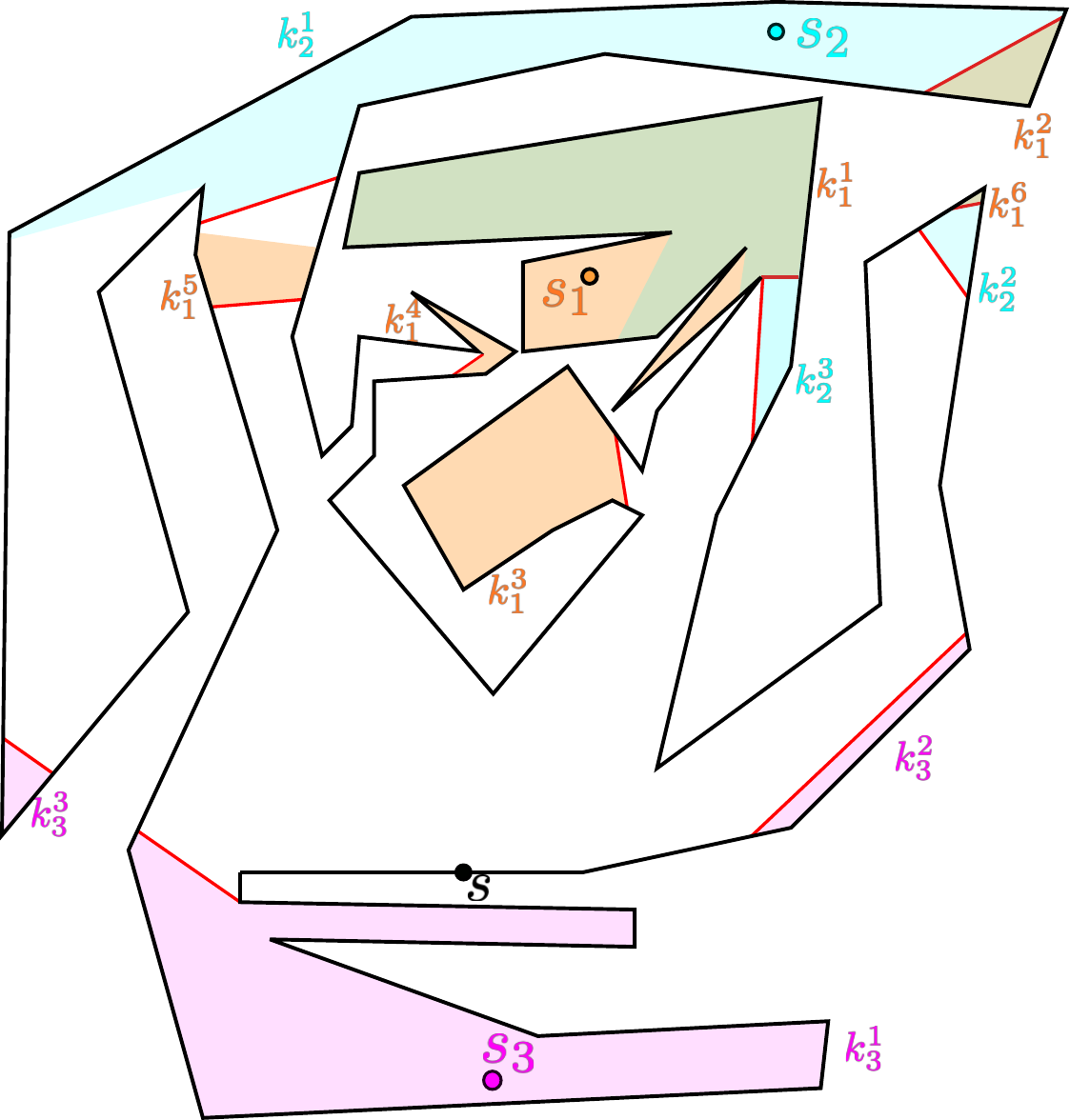}{(a)}\hfill
\comic{.48\textwidth}{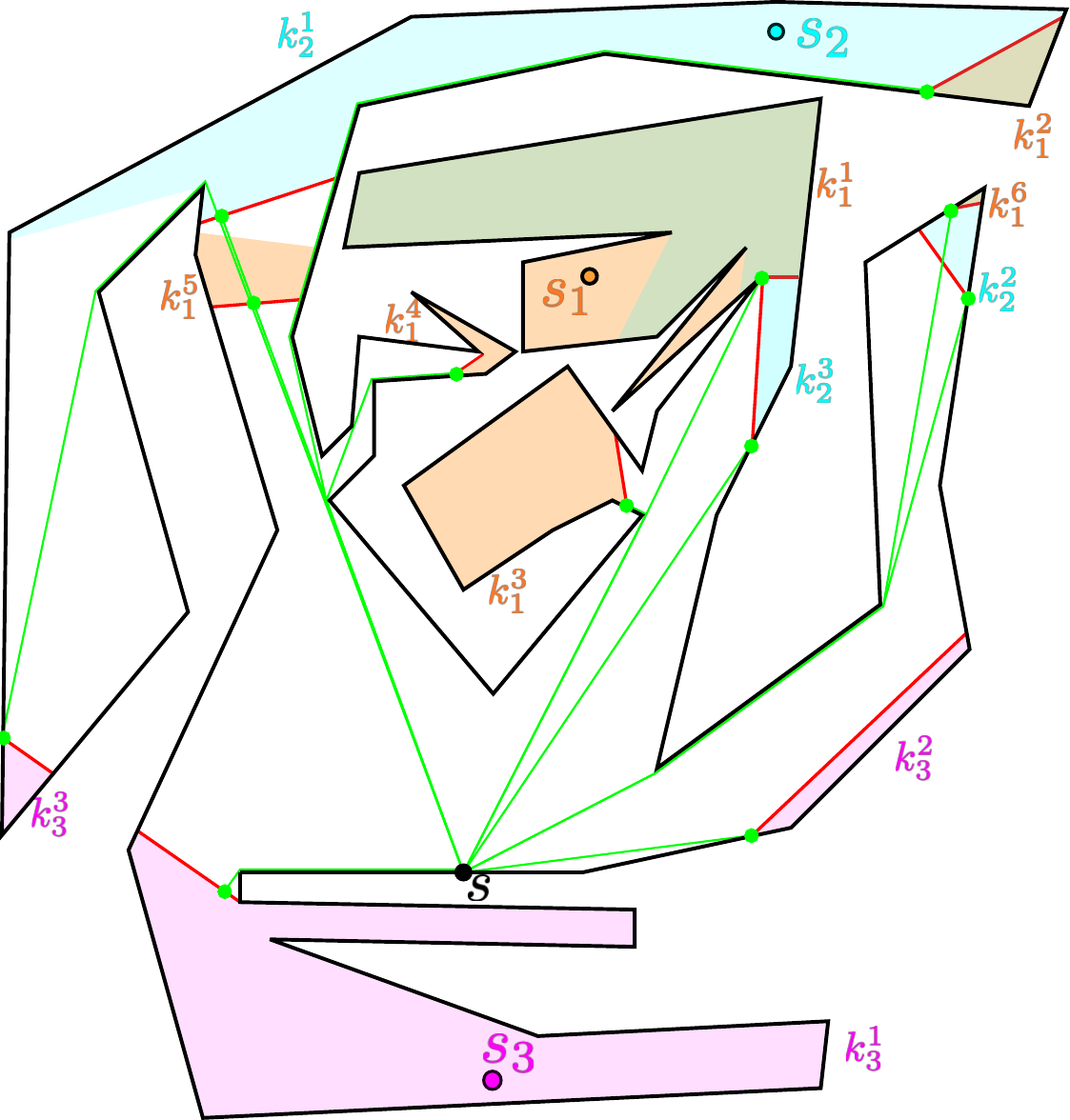}{(b)}\\
\comic{.48\textwidth}{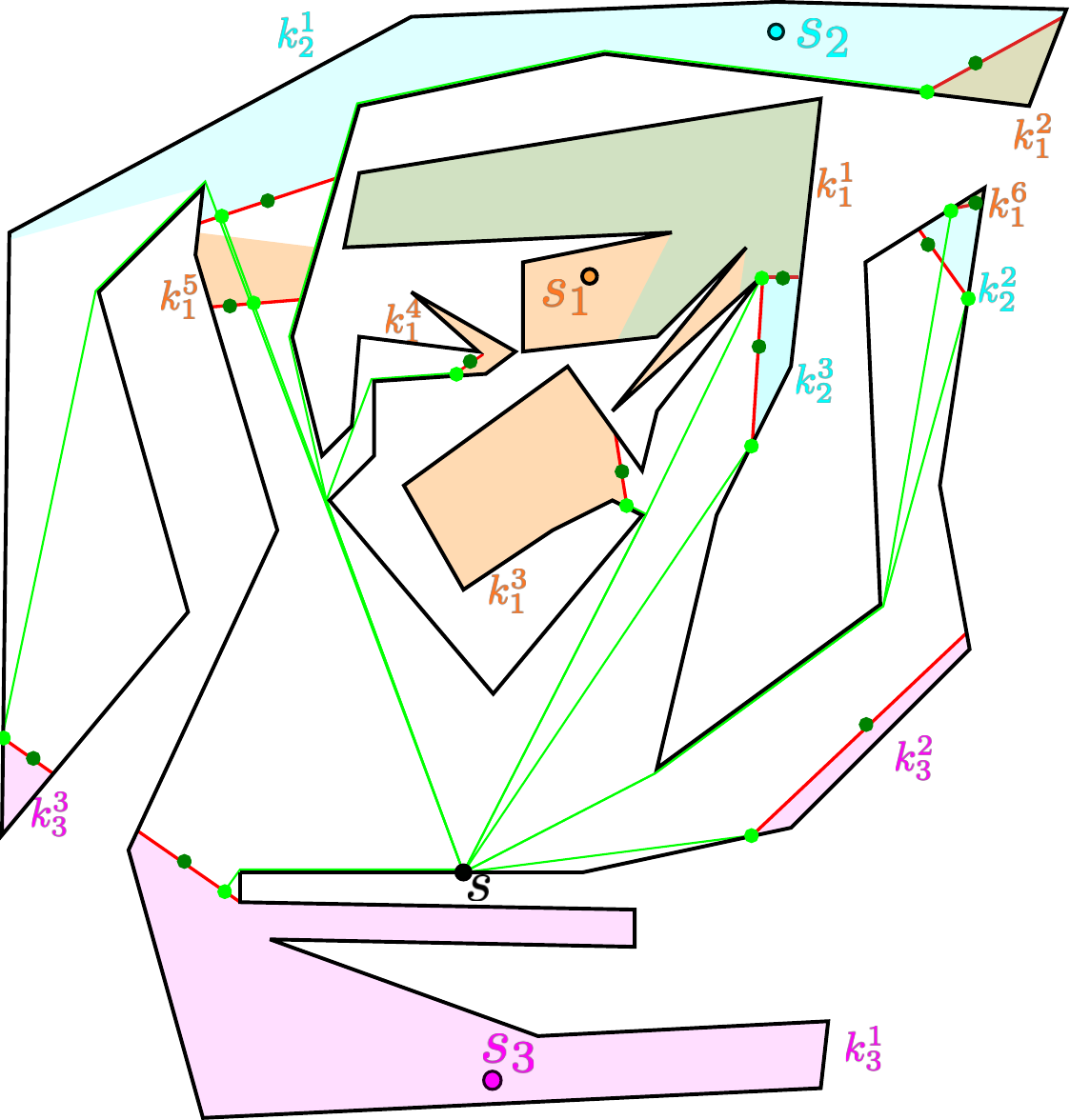}{(c)}\hfill
\comic{.5\textwidth}{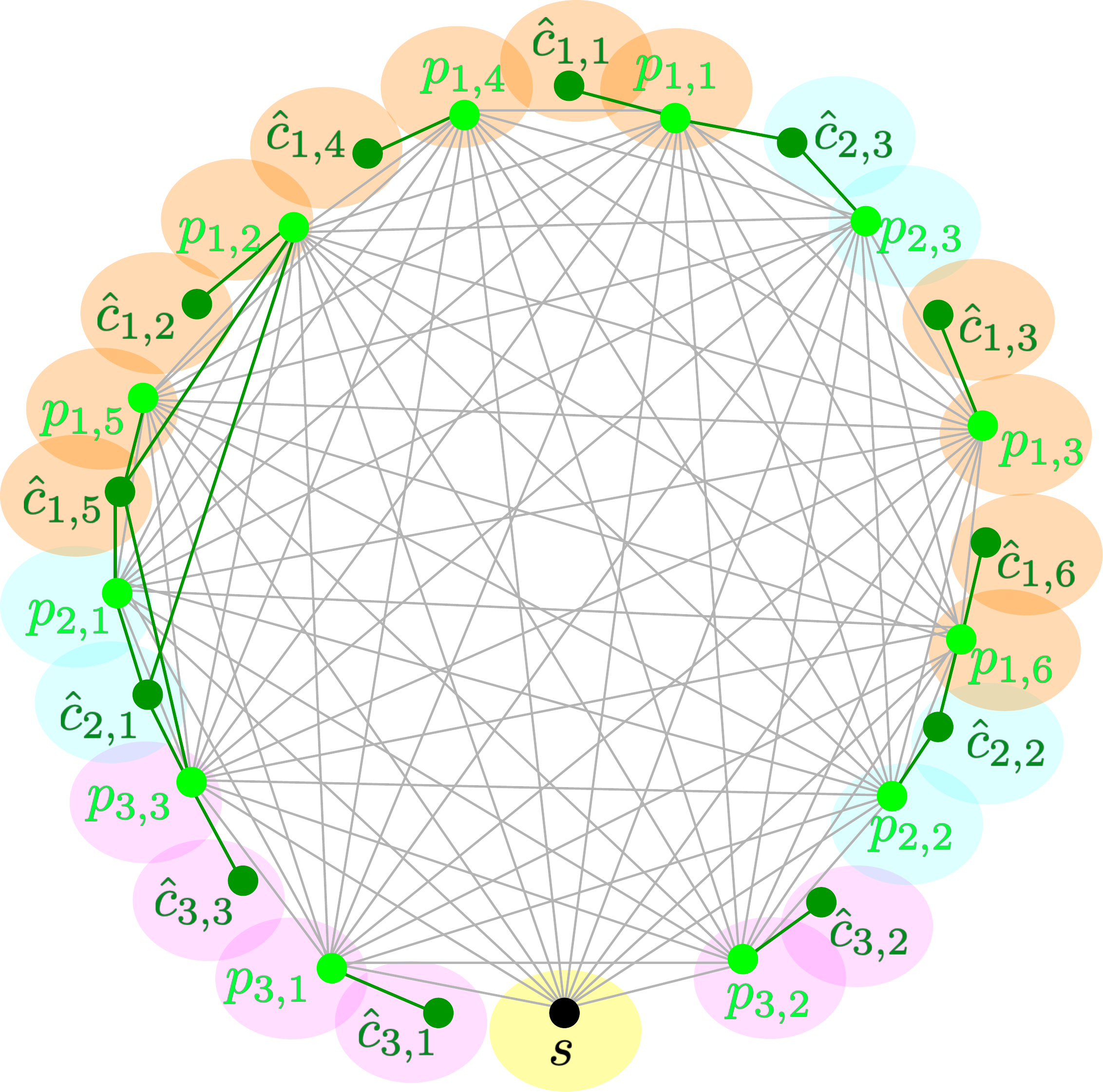}{(d)}
  \caption{\small Example for the idea of our approximation algorithm, $S=\{s_1, s_2, s_3\}$. (a)-(c): The cuts $c_{i,j}$ are shown in red, geodesics and all $p_{i,j}$ are shown in light green, all $\hat{c}_{i,j}$ are shown in dark green. The CCs of the visibility region of a point $s_i$ are colored in a lighter shade of the same color as the point itself (orange for $s_1$, turquoise for $s_2$, and pink for $s_3$). Line segments are slightly offset to enhance visibility in case they coincide with polygon boundary. (d) Resulting graph $G$. Edges with edge cost $0$ are shown in dark green, edges with edge cost of the length of the geodesic between the two points in $P$ are shown in gray. We highlight each vertex of $G$ with the color of the point $s_i$ to which it belongs. All colored in the same color constitute the set $\gamma_i$ ($\gamma_1$ highlighted in light orange, $\gamma_2$ highlighted in light turquoise, $\gamma_3$ highlighted in light pink, $\gamma_0$ highlighted in yellow).}
  \label{fig:appx-idea}
\end{figure}

The basic idea of our approximation algorithm is to create a candidate point for each connected component of the $k$-visibility region of each point in $S$. These candidate points are defined by the intersections of geodesics from the starting point $s$ to the cuts and the cuts themselves. We then build a complete graph on these candidate points (using the length of geodesics in $P$ between two points as the edge length in the graph). Finally, we group all candidate points that belong to the same point in $S$ and construct a group Steiner tree; by doubling this tree, we obtain a route. Our approximation algorithm performs the following steps:
 
\begin{enumerate}
\item For each $s_i\in S$, we compute the $k$-visibility region within $P$, $\kVR(s_i)$---and say that all CCs $\kVR^j(s_i)$ have ``color'' $s_i$. We denote $\kVR^j(s_i)$ as $k^j_i$. Let the cut of each $k^j_i$ be denoted by $c_{i,j}$, and let $\mathcal{C}^{all}$ denote the set of all cuts in $P$. See Figure~\ref{fig:appx-idea}(a) for an example of this step. 

\item We compute a geodesic $g_{i,j}$ from $s$ to each cut $c_{i,j}$. Let $p_{i,j}$ be the point where $g_{i,j}$ intersects $c_{i,j}$.  See Figure~\ref{fig:appx-idea}(b) for an example of this step, the $p_{i,j}$ are shown in light green.

\item We build the complete graph on the $p_{i,j}$ and $s$: for an edge $\{x,y\}$, we have cost$(\{x,y\})=\mbox{geodesic}_P(x,y)$. We introduce further vertices and edges: one vertex $\hat{c}_{i,j}$ per cut $c_{i,j}\in\mathcal{C}^{all}$. We add edges $\{p_{i,j},\hat{c}_{i,j}\}$ with edge cost $0$, and edges $\{p_{i,j},\hat{c}_{i',j'}\}$ with edge cost $0$ for all cuts $c_{i',j'}$ that $g_{i,j}$ intersects. (Rationale: any path or tour visiting $p_{i,j}$ must visit $c_{i',j'}$/$c_{i,j}$.)
Let the resulting graph be denoted as $G=(V,E)$. See Figure~\ref{fig:appx-idea}(c)/(d): the points of type $\hat{c}_{i,j}$ are shown in green. 
We have $|V(G)|=O(n\cdot |S|)$.

\item With $\gamma_i=\bigcup_{j=1}^{J_i} p_{i,j} \cup \bigcup_{j=1}^{J_i} \hat{c}_{i,j}$, $\gamma_{0}=s$, $Q=|S|+1$---that is, each group $\gamma_i$ contains all vertices in $V(G)$ of  color $s_i$, $\gamma_0$ contains the starting point that we must visit---we approximate the group Steiner tree problem on $G$, using the 
approximation by Garg et al.~\cite{gkr-paags-00}, the approximation ratio is $O\big(f(|V(G)|,|S|)\big)$, where $f(N,M)=\log^2 N \log\log N \log M$, e.g., polylogarithmic in $|V(G)|$ and~$|S|$.

\item We double the resulting tree to obtain a route $R$; it visits at least one vertex per color (one point in each $\kVR(s_i)$). Thus, $R$ is a feasible solution for $k$-TrWRP($S,P,s$) visiting one 
point per $\gamma_i$. $R$ is a polylog-approximation to the best tour that is feasible for $k$-TrWRP($S,P,s$), visiting one 
point per $\gamma_i$ using edges in $G$ (denoted by OPT$_G$($S,P,s$)). 
\end{enumerate}

To prove that $R$ is indeed an approximation with the claimed approximation factor, we alter the optimum $k$-transmitter watchman route, OPT($S,P,s$), (which we of course do not know in reality) to pass 
points that represent vertices of $V(G)$, and show that this new tour is at most $3$ times as long as the optimum route.
The visited points are intersection points of independent geodesics and cuts (we obtain independent geodesics by ordering the geodesics to essential cuts by non-increasing length and filtering out geodesics to cuts that were visited by longer geodesics).
 The basic idea is:
\begin{enumerate}
\item[a.] We identify all cuts of the $\kVR(s_i)$ that OPT($S,P,s$) visits, let these be the set $\mathcal{C}$ ($\mathcal{C}\subseteq\mathcal{C}^{all}$). 
Let $o_{i,j}$ denote the point where OPT($S,P,s$) visits $c_{i,j}$ (for the first time).

\item[b.] We identify the subset of essential cuts $\mathcal{C'}\subseteq\mathcal{C}$.

\item[c.] We order the geodesics to the essential cuts $\mathcal{C'}$ by decreasing length: $\ell(g_1)\geq\ell(g_2)\geq\ldots\geq\ell(g_{|\mathcal{C'}|})$, where $\ell(\cdot)$ defines the Euclidean length.

\item[d.] $\mathcal{C''}\leftarrow\mathcal{C'}$; FOR $t=1$ TO $|\mathcal{C'}|$, we identify all $\mathcal{C}_t\subset\mathcal{C'}$ that $g_t$ intersects, and set $\mathcal{C''}\leftarrow\mathcal{C''}\setminus\mathcal{C}_t$.\\ 
$\mathcal{C''}\subseteq\mathcal{C'}$. We let $\mathcal{G}_{\mathcal{C''}}$ be the set of geodesics that end at cuts in $\mathcal{C''}$.

\item[e.] The geodesics in $\mathcal{G}_{\mathcal{C''}}$ constitute a set of {\it independent} geodesics, that is, no essential cut is visited by two of these geodesics. Moreover, each essential cut visited by OPT($S,P,s$)---each cut in $\mathcal{C'}$---is touched by exactly one of the geodesics. 

\item[f.] The geodesics in $\mathcal{G}_{\mathcal{C''}}$ intersect the cuts in $\mathcal{C''}$ in points of the type $p_{i,j}$, points that represent vertices of $V(G)$. We denote the set of all these 
points as $\mathcal{P}_{\mathcal{C''}}$ ($\mathcal{P}_{\mathcal{C''}}\subseteq\{p_{i,j}\mid i=1,\ldots,|S|, j=1,\ldots, J_i\}$).

\item[g.] We build the relative convex hull of all $o_{i,j}$ and all points in $\mathcal{P}_{\mathcal{C''}}$ (relative w.r.t. the polygon $P$). We denote this relative convex hull by CH$_{P}(\mbox{OPT},\mathcal{P}_{\mathcal{C''}})$.

\item[h.] Because we have a set of independent geodesics, no geodesic can intersect CH$_{P}(\mbox{OPT},\mathcal{P}_{\mathcal{C''}})$ between a point $o_{i,j}$ and a point $p_{i,j}$ on the same cut. 
Thus, between any pair of 
points of the type $o_{i,j}$ on CH$_{P}(\mbox{OPT},\mathcal{P}_{\mathcal{C''}})$, we have at most two 
points of $\mathcal{P}_{\mathcal{C''}}$.
We show that CH$_{P}(\mbox{OPT},\mathcal{P}_{\mathcal{C''}})$ has length of at most three times $\|$OPT($S,P,s$)$\|$. 

\item[i.] The relative convex hull of the points in $\mathcal{P}_{\mathcal{C''}}$, CH$_{P}(\mathcal{P}_{\mathcal{C''}})$, is not longer than CH$_{P}(\mbox{OPT},\mathcal{P}_{\mathcal{C''}})$, and we show that CH$_{P}(\mathcal{P}_{\mathcal{C''}})$ visits 
one point per $\gamma_i$ (except for $\gamma_0$).

\item[j.] Because $s$ ($=\gamma_0$) might be located in the interior of CH$_{P}(\mathcal{P}_{\mathcal{C''}})$, we need to connect $s$ to CH$_{P}(\mathcal{P}_{\mathcal{C''}})$. This costs at most $\|\mbox{OPT}(S,P,s)\|$. 

\item[k.] Thus, we obtain (note $f(N,M)=\log^2  N$$ \log\log N \log M$):
\begin{align}
\|R\| 
& \leq
\alpha_1\cdot f(|V(G)|,|S|)\|\mbox{OPT}_G(S,P,s)\|
\leq
\alpha_2\cdot f(n|S|,|S|)\|\mbox{CH}_{P}(\mathcal{P}_{\mathcal{C''}})\| \nonumber\\
& \leq
\alpha_3\cdot f(n|S|,|S|)\|\mbox{CH}_{P}(\mbox{OPT},\mathcal{P}_{\mathcal{C''}})\|
\leq
\alpha_4\cdot f(n|S|,|S|)\|\mbox{OPT}(S,P,s)\|\nonumber
\end{align}
for suitable constants $\alpha_1,\ldots, \alpha_4$.
\end{enumerate}

Hence, to show Theorem~\ref{th:appx-dis}, we need to prove steps e, h, and i. We show step e using Lemma~\ref{le:indep}; 
step h using Lemmas~\ref{le:intersecting-ch},~\ref{le:max2}, and~\ref{le:ch-vs-opt}; and step i using Lemmas~\ref{le:compare-chs},~\ref{le:visit-PC''},~\ref{le:visit-all-cuts-in-C}, and~\ref{le:visit-gammais}.  

\begin{lemma}\label{le:indep}
For the geodesics in $\mathcal{G}_{\mathcal{C''}}$,  we have:
\begin{enumerate}
\item[I.] The geodesics in $\mathcal{G}_{\mathcal{C''}}$ are independent, that is, no cut in $\mathcal{C'}$ is visited by two of these geodesics. 
\item[II.] Each cut in $\mathcal{C'}$ is visited by a geodesic in $\mathcal{G}_{\mathcal{C''}}$.
\end{enumerate}
\end{lemma}
\begin{proof}
In step c, we order the geodesics to the essential cuts $\mathcal{C'}$ by decreasing length: $\ell(g_1)\geq\ell(g_2)\geq\ldots\geq\ell(g_{|\mathcal{C'}|})$. In step d, we iterate over these geodesics in the order $g_1, g_2, \ldots, g_{|\mathcal{C'}|}$: if the current geodesic $g_t$ intersects cuts $c_{t_1},\ldots,c_{t_Y}\in\mathcal{C'}$ we delete the (shorter) geodesics to these cuts ($g_{t_1},\ldots, g_{t_Y}$). Thus, after the last iteration, no two geodesics of those we are left with---the geodesics in $\mathcal{G}_{\mathcal{C''}}$---visit the same cut in $\mathcal{C'}$. 

Moreover, $g_1, g_2, \ldots, g_{|\mathcal{C'}|}$ visit all cuts in $\mathcal{C'}$. We only delete a geodesic from this set if its cut is already visited by a longer geodesic. Thus, we maintain the property that all cuts in $\mathcal{C'}$ are visited.
\end{proof}

\begin{lemma}\label{le:intersecting-ch}
Consider a cut $c\in\mathcal{C''}$\!, from CC $j$ of a $k$-visibility region for $s_i\in S$, $\kVR^{j}(s_{i})$, for which both the point $o_{i,j}$ and the point $p_{i,j}$ are on {\rm CH$_{P}(\mbox{OPT},\mathcal{P}_{\mathcal{C''}})$}. No geodesic in $\mathcal{G}_{\mathcal{C''}}$ intersects $c$ between $o_{i,j}$ and $p_{i,j}$.
\end{lemma}
\begin{proof}
Assume that there exists a geodesic $g_{c'}\in\mathcal{G}_{\mathcal{C''}}$ to a cut $c'\neq c, c'\in\mathcal{C''}$ that intersects $c$ between $o_{i,j}$ and $p_{i,j}$.
Let $c'$ be the cut of $\kVR^{j'}(s_{i'})$. Let $p_c$ denote the point in which $g_{c'}$ intersects $c$. 
If $\ell(g_{c'})>\ell(g_{c})$, we would have deleted $g_c$ in step d, hence $c\notin\mathcal{C''}$. If $\ell(g_{c'})<\ell(g_{c})$, the geodesic to $c'$ restricted to the part between $s$ and $p_{c}$, $g_{c'[s;p_{c}]}$, is shorter than $g_c$, a contradiction to $g_c$ being the geodesic to $c$. If $\ell(g_{c'})=\ell(g_{c})$, either $\ell(g_{c'[s;p_{c}]})<\ell(g_{c'})=\ell(g_{c})$ or (if $p_{c}$ on $c'$) $p_{i,j}=p_{c}$ and the claim holds.
\end{proof}

\begin{lemma}\label{le:max2}
Between any pair of points of the type $o_{i,j}$ on {\rm CH$_{P}(\mbox{OPT},\mathcal{P}_{\mathcal{C''}})$}, we have at most two points in $\mathcal{P}_{\mathcal{C''}}$.
\end{lemma}

\begin{figure}[b]   
\centering
\comicII{.6\textwidth}{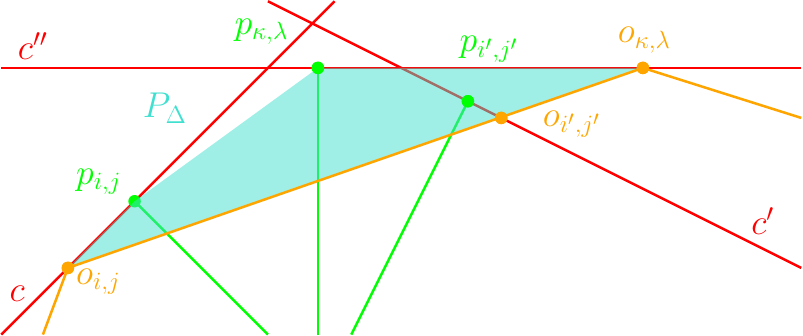}
  \caption{\small Example for the proof of Lemma~\ref{le:max2}. Cuts are shown in red, points of the type $p_{i,j}$ in green, and the optimal route and points of the type $o_{i,j}$ in orange. The convex polygon $P_\triangle$ is shown in turquoise.  }
  \label{fig:at-most-two}
\end{figure}

\begin{proof}
Let $o_{i,j}$ and $o_{i', j'}$ be two consecutive points from OPT on CH$_{P}(\mbox{OPT},\mathcal{P}_{\mathcal{C''}})$, see Figure~\ref{fig:at-most-two} for an example of this proof construction.
By Lemma~\ref{le:intersecting-ch}, $p_{i,j}$ and $p_{i', j'}$ can lie between $o_{i,j}$ and $o_{i', j'}$, but we can have no point 
$p_{\kappa,\lambda}$ between $o_{i,j}$ and $p_{i,j}$ or between $o_{i',j'}$ and $p_{i',j'}$.
Assume that there exists a point $p_{\kappa,\lambda}$ between $p_{i,j}$ and $p_{i',j'}$ on CH$_{P}(\mbox{OPT},\mathcal{P}_{\mathcal{C''}})$. Moreover, let $p_{i,j}$, $p_{i',j'}$, and $p_{\kappa,\lambda}$ be on cuts $c, c'$ and $c''$, respectively. OPT visits $o_{\kappa,\lambda}$ on $c''$. As $o_{i,j}$ and $o_{i', j'}$ are consecutive points from OPT on CH$_{P}(\mbox{OPT},\mathcal{P}_{\mathcal{C''}})$, OPT visits the three points either in order $o_{i,j}, o_{i', j'}, o_{\kappa,\lambda}$ or $o_{\kappa,\lambda}, o_{i,j}, o_{i', j'}$. W.l.o.g., assume the order  $o_{i,j}, o_{i', j'}, o_{\kappa,\lambda}$. The cut $c''$ is a straight-line segment. Consider the convex polygon $P_\triangle$ with vertices $o_{i,j}, p_{i,j}, p_{\kappa, \lambda}, o_{\kappa, \lambda}, o_{i',j'}, o_{i,j}$. The point $p_{i',j'}$ must lie in $P_\triangle$'s interior. Moreover, $o_{i',j'}$ cannot lie on CH$_{P}(\mbox{OPT},\mathcal{P}_{\mathcal{C''}})$; a contradiction.
\end{proof}

\begin{lemma}\label{le:ch-vs-opt}
$\|\mbox{\rm CH}_{P}(\mbox{OPT},\mathcal{P}_{\mathcal{C''}})\| \leq   3 \cdot\|\mbox{\rm OPT}(S,P,s)\|$
\end{lemma}
\begin{proof}
By Lemmas~\ref{le:intersecting-ch} and~\ref{le:max2}, we have that between two consecutive points of $\mbox{OPT}(S,P,s)$ on $\mbox{CH}_{P}(\mbox{OPT},\mathcal{P}_{\mathcal{C''}})$, $o_{i,j}$ and $o_{i',j'}$, we have at most two points where a geodesic visits a cut: $p_{i,j}$ and $p_{i',j'}$.
At $p_{i,j}$, the geodesic $g_{i,j}$ to the cut $c_{i,j}$ of $\kVR^j(s_i)$ visits $c_{i,j}$. Because $o_{i,j}$ and $p_{i,j}$ are both on $c_{i,j}$ (and $o_{i',j'}$ and $p_{i',j'}$ are both on $c_{i',j'}$), $g_{i,j}$ intersects $\mbox{OPT}(S,P,s)$ between $o_{i,j}$ and $o_{i',j'}$. Let the point of intersection be denoted as $\rho_{i,j}$. Because $g_{i,j}$ is a geodesic, we have: $\ell(\rho_{i,j}, p_{i,j})\leq \ell(\rho_{i,j}, o_{i,j})$. Analogously, we have: $\ell(\rho_{i',j'}, p_{i',j'})\leq \ell(\rho_{i',j'}, o_{i',j'})$. See Figure~\ref{fig:reroute-opt} for an example.

Hence, if we alter $\mbox{OPT}(S,P,s)$ between $o_{i,j}$ and $o_{i',j'}$ to follow the sequence $o_{i,j},\rho_{i,j}, p_{i,j},$ $\rho_{i,j}, \rho_{i',j'}, p_{i',j'}, \rho_{i',j'}, o_{i',j'}$, we obtain a new tour $T$ that visits all points on $\mbox{CH}_{P}(\mbox{OPT},\mathcal{P}_{\mathcal{C''}})$ and $\|T\|\leq 3 \cdot\|\mbox{OPT}(S,P,s)\|$, since $\ell(\rho_{i,j}, p_{i,j})\leq \ell(\rho_{i,j}, o_{i,j})$ for all pairs $i,j$. $\mbox{CH}_{P}(\mbox{OPT},\mathcal{P}_{\mathcal{C''}})$ is the shortest tour that visits all these points, thus, $\|\mbox{CH}_{P}(\mbox{OPT},\mathcal{P}_{\mathcal{C''}})\|\leq\|T\|$. Hence, $\|\mbox{CH}_{P}(\mbox{OPT},\mathcal{P}_{\mathcal{C''}})\| \leq   3 \cdot\|\mbox{OPT}(S,P,s)\|$. \end{proof}

\begin{figure}[t]
\centering
\comic{.4\textwidth}{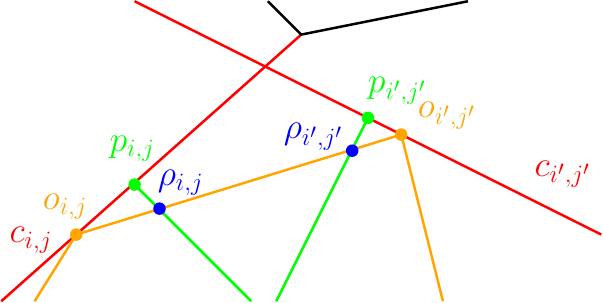}{(a)}\hfill
\comic{.4\textwidth}{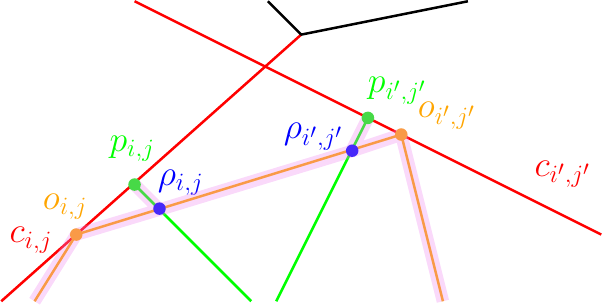}{(b)}
  \caption{\small  The cuts $c_{i,j}$ and $c_{i',j'}$ are shown in red; $o_{i,j}$, $o_{i',j'}$ and $\mbox{OPT}(S,P,s)$ are shown in orange; $p_{i,j}$, $p_{i',j'}$, $g_{i,j}$, and $g_{i',j'}$ are shown in green; $\rho_{i,j}$ and $\rho_{i',j'}$ are shown in blue; and $T$ is shown in pink. A part of $P$'s boundary is shown in black. }
  \label{fig:reroute-opt}
\end{figure}

\begin{lemma}\label{le:compare-chs}
$\|\mbox{\rm CH}_{P}(\mathcal{P}_{\mathcal{C''}})\|\leq \|\mbox{\rm CH}_{P}(\mbox{\rm OPT},\mathcal{P}_{\mathcal{C''}})\|$.
\end{lemma}
\begin{proof}
We have $\mathcal{P}_{\mathcal{C''}}\subseteq\mbox{OPT}\cup\mathcal{P}_{\mathcal{C''}}$, hence, the claim follows trivially.
\end{proof}

\begin{lemma}\label{le:visit-PC''}
All points in $\mathcal{P}_{\mathcal{C''}}$ lie on their relative convex hull\/
$\mbox{\rm CH}_{P}(\mathcal{P}_{\mathcal{C''}})$.
\end{lemma}
\begin{proof}
Assume there is a point $p_{i,j}\in\mathcal{P}_{\mathcal{C''}}, p_{i,j}\notin\mbox{CH}_{P}(\mathcal{P}_{\mathcal{C''}})$. By Lemma~\ref{le:indep}, there exists no $g_c\in\mathcal{G}_{\mathcal{C''}}$ that intersects $c_{i,j}$. The cut $c_{i,j}$ connects two points $x,y\in\partial(P)$, and we have $\mbox{CH}_{P}(\mathcal{P}_{\mathcal{C''}})\subseteq P_s(c_{i,j})$ (that is, $\mbox{CH}_{P}(\mathcal{P}_{\mathcal{C''}})$ does not cross $c_{i,j}$). Thus, $p_{i,j}$ must lie on $\mbox{CH}_{P}(\mathcal{P}_{\mathcal{C''}})$; a contradiction. See Figure~\ref{fig:all-on-CH} for an example.
\end{proof}

\begin{figure}[]
\centering
\comicII{.6\textwidth}{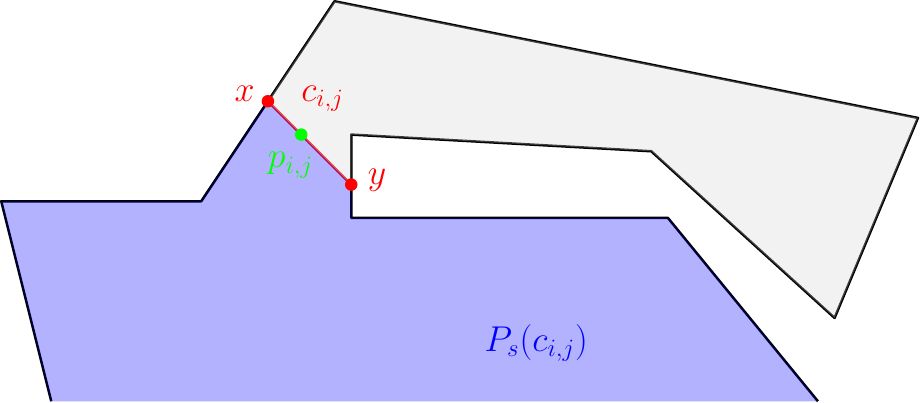}
  \caption{\small For a point $p_{i,j}\in\mathcal{P}_{\mathcal{C''}}$, no geodesic from $\mathcal{G}_{\mathcal{C''}}$ intersects $c_{i,j}$. The cut $c_{i,j}$ connects two points $x,y\in\partial(P)$. Because $\mbox{CH}_{P}(\mathcal{P}_{\mathcal{C''}})\subseteq P_s(c_{i,j})$ (with $P_s(c_{i,j})$ shown in blue), $p_{i,j}$ must lie on~$\mbox{CH}_{P}(\mathcal{P}_{\mathcal{C''}})$.}
  \label{fig:all-on-CH}
\end{figure}

\begin{lemma}\label{le:visit-all-cuts-in-C}
$\mbox{\rm CH}_{P}(\mathcal{P}_{\mathcal{C''}})$ visits all cuts in $\mathcal{C}$. 
\end{lemma}
\begin{proof}
We have $\mathcal{C}=\mathcal{C''}\cup\{\mathcal{C'}\setminus\mathcal{C''}\}\cup\{\mathcal{C}\setminus\mathcal{C'}\}$. Cuts in $\{\mathcal{C}\setminus\mathcal{C'}\}$ are dominated by cuts in $\mathcal{C'}$. Thus, any tour visiting all cuts in $\mathcal{C'}$ must visit all cuts in $\{\mathcal{C}\setminus\mathcal{C'}\}$. Cuts in $\mathcal{C''}$ are visited by Lemma~\ref{le:visit-PC''}. 

Assume that there is a cut $c\in\{\mathcal{C'}\setminus\mathcal{C''}\}$ not visited by $\mbox{CH}_{P}(\mathcal{P}_{\mathcal{C''}})$.  Then the cut $c$ is not in $\mathcal{C''}$ (we filtered $g_c$ out in step d), hence, there exists a geodesic $g_{c_{i,j}}\in\mathcal{G}_{\mathcal{C''}}$ with $\ell(g_{c_{i,j}})\geq \ell(g_c)$ that intersects $c$; $g_{c_{i,j}}$ visits $c_{i,j}$ in the point $p_{i,j}$. Thus, any tour that visits both $s$ and $p_{i,j}$ must intersect $c$. By Lemma~\ref{le:visit-PC''}, $\mbox{CH}_{P}(\mathcal{P}_{\mathcal{C''}})$ is such a tour; a contradiction.
\end{proof}


\begin{lemma}\label{le:visit-gammais}
$\mbox{\rm CH}_{P}(\mathcal{P}_{\mathcal{C''}})$ visits one point per $\gamma_i$, except for $\gamma_0$.
\end{lemma}
\begin{proof}
Because OPT($S,P,s$) is feasible, the set $\mathcal{C}$ does include at least one cut colored in $s_i, \forall i$.  By Lemma~\ref{le:visit-all-cuts-in-C}, $\mbox{CH}_{P}(\mathcal{P}_{\mathcal{C''}})$ visits all cuts in $\mathcal{C}$. Hence, it visits at least one point per~$\gamma_i$.
\end{proof}
This concludes the proof of Theorem~\ref{th:appx-dis}.

\section{Conclusion}\label{sec:concl}

We proved that even in simple polygons the $k$-transmitter watchman route problem for $S\subset P$ cannot be approximated to within a logarithmic factor (unless P=NP) --- both the variant with a given starting point and the floating watchman route. Moreover, we provided an approximation algorithm for $k$-TrWRP($S,P,s$), that is, the variant where we need to see a discrete set of points $S\subset P$ with a given starting point. The approximation ratio of our algorithm is $O(\log^2(|S|\cdot n) \log\log (|S|\cdot n) \log|S|)$.

Obvious open questions concern approximation algorithms for the other versions of the WRP for $k$-transmitters: $k$-TrWRP($P,P,s$),  $k$-TrWRP($P,P$), and $k$-TrWRP($S,P$). Moreover, for $0$-transmitters (``normal'' guards), we have a very clear structure: any watchman route must visit all non-dominated extensions of edges incident to a reflex vertex. Any structural analogue for $k$-transmitters ($k\geq2$) would be of great interest.  

\bibliographystyle{is-unsrt}
\bibliography{lit}

\end{document}